\documentclass{elsart}

\usepackage{amssymb}
\usepackage[centertags]{amsmath}
\usepackage{graphicx}
\usepackage{colortbl}
\usepackage[utf8]{inputenc}
\makeatother

\newtheorem{theorem}{Theorem}[section]

\newtheorem{definition}[theorem]{Definition}

\newenvironment{proof}{\removelastskip\par\medskip}
{\penalty-20\null\hfill$\square$\par\medbreak}

\begin{document}

\begin{frontmatter}



\title{Model selection of stochastic simulation algorithm 
 based on generalized \\ divergence measures}
\thanks[talk]{ This research was supported, in part, by grants from 
Universit\'e Cheikh Anta Diop, Dakar (S\'en\'egal). We are grateful to many seminar
participants and to anonymous referees for comments.}

\author{\small{ Papa NGOM  and  B. Don Bosco  DIATTA}}

\address{LMA-Laboratoire de Math\'ematiques Appliqu\'ees \\
 Universit\'e Cheikh Anta Diop BP 5005 Dakar-Fann S\'en\'egal\\
  e-mail : pngom@ucad.sn }

\begin{abstract}
\par \hspace{1cm}MCMC methods (Monte Carlo Markov Chain) are a class of methods used to perform simulations per a probability distribution $P$. These methods are often used when we have difficulties to directly sample per a given probability distribution $P$ . This distribution is then considered as a target and generates a Markov chain $(X_n)_{n\in\mathbb{N}}$ that, when $n$ is large we have $X_n\sim P$. These MCMC methods consist of several simulation strategies including the \emph{Independent Sampler (IS)}, the \emph{Random Walk of Metropolis Hastings \small{(RWMH)}}, the \emph{Gibbs sampler}, the \emph{Adaptive Metropolis (AM)} and \emph{Metropolis Within Gibbs (MWG)} strategy. Each of these strategies can generate a Markov chain and is associated with a convergence speed.
It is interesting, with a given target law, to compare several simulation strategies for determining the best. Chauveau and Vandekerkhove \cite{Chauv2007} have compared  IS and RWMH strategies  using the  Kullback-Leibler divergence measure. In our article we will compare our five simulation methods already mentioned using generalized divergence measures. These divergence measures are taken in family of $\alpha$-divergence measures \cite{Cichocki2010} , with a parameter $\alpha$. This is the R\'enyi divergence, Tsallis divergence and $D_\alpha$ divergence .  
\end{abstract}

\begin{keyword}
   MCMC methods, Metropolis-Hastings algorithm, Gibbs sampler, Adaptive Metropolis, Metropolis within Gibbs, simulation strategy, target density, proposal density, $\alpha$-divergence.\\
  AMS Subject Classification : 60J60, 62F03, 62F05, 94A17.
\end{keyword}
\end{frontmatter}
\section{Introduction}
 \par \hspace{1cm} In many areas of science, computer is essential. So computers are used for testing, calculations, simulations ... Beyond performance, the computation time is very important in the choice of calculation methods. In the field of statistical, the simulation of random variables are required to perform  assessments such integral calculation. Thus we resort to numerical analysis methods or probabilistic methods. It is the latter that concern us in particular.  The integral form is rewritten like an expectation $\mathbb{E}(X)$ and use the \emph{law of large numbers} for the Monte Carlo integration and the \emph{ergodic theorem} for MCMC methods. We will focus here on MCMC methods with which one performs sampling according to density $ f $ via Markov chains. Among MCMC methods we have Metropolis-Hastings algorithms, Gibbs sampler and others adaptive and hybrid algorithms. 
\par \hspace{1cm}Thus the major challenge of these algorithms, like many computational techniques, is the saving time i.e. their speedy convergence to the stationary distribution with density $f$. This is due to the fact that the convergence time is a brand of efficiency. This paper investigates the problem of selecting a good strategy simulation. We have several simulation strategies that are \emph{Independence Sampler}, \emph{Random Walk of Metropolis-Hastings} \cite{Millet}, \emph{Gibbs sampler}, \emph{Metropolis Within Gibbs} \cite{Latusz2013}, \emph{Adaptive Metropolis} \cite{Haario2001}...
 The problem for all these strategies is the time they take to the generated a Markov chain which converge to the stationary distribution with  probability density function $f$.
\par \hspace{1cm}The objective of this paper is to compare different strategies for simulation in order to select the best.  Chauveau and Vandekerkhove \cite{Chauv2007} have given in their paper a methodology for selecting the best of $k$ candidate strategies with the Kullback-Leibler divergence. We use in our article generalized divergence measures that are $\alpha$-divergence \cite{Poczos2011} to compare simulation strategies. 
These divergence measures that we use are \emph{$\alpha$-divergence} $D_\alpha$, the \emph{R\'enyi $\alpha$-divergence} $R_\alpha$, the \emph{Tsallis $\alpha$-divergence} $T_\alpha$. We know that these divergences are all generalizations of the Kullback-Leibler divergence. We have the Kullback-Leibler divergence when the parameter $\alpha$ of these divergence measures tends to 1. If we have two simulation strategies $s_1$ and $s_2$ of respective probability densities 
$p^n_1$\footnote{$p_1^n$ is density function of $X_n$ of Markov chain generated by strategy $s_1$} and $p_2^n$ at time $n$ and a target density $f$ which we want to have samples we can use these divergence measures to compare $s_1$ and $s_2$. Indeed, if we consider for example the R\'enyi divergence, we can compare $R_{\alpha}(p_1^n, f)$   and $R_{\alpha}(p_2^n, f)$, for each iteration $n$, to see which of $s_1$ and $s_2$ is the best.
\par \hspace{1cm}We will show in Section 2 these three divergence measures. We know that they belong to the family of Csisz\'ar $\phi$-divergences. In Section 3 we will show two important results that are convergence to 0 of these divergences and building estimator for each divergence. We will show in Section 4 the asymptotic distributions of our estimators. Then  we have some application examples in order to illustrate our methodology (Section 5). Finally, we do  discussion in Section 6 to explain our various examples and see the lesson we can draw threreof.

\section{Description of divergence measures}
 \par \hspace{0.5cm} Let $(\cal X, \cal A,\lambda)$ be an arbitrary measure space with
 $\lambda$ being a finite or $\sigma-$finite measure. Let also
 $\mu_1,\mu_2$ probability measures on $ \cal X$ such that
 $\mu_1,\mu_2 \ll \lambda$ (absolutely continuous).\\
 Denote the Randon-Nikodym derivatives (densities) of $\mu_i$ with
 respect to $\lambda$ by $p_i(x)$:
 $$ p_i(x)=\frac{\mu_i(dx)}{\lambda(dx)}, \quad i=1,2. $$
 
 


\begin{definition} \label{h1}
 Kullback-Leibler’s relative divergence (also called relative entropy) between two probability measures $\mu_1$ , $\mu_2$ is defined by

\begin{equation}
  K(\mu_1,\mu_2)=\int_{\cal X}p_1(x)\log \big(
 \frac{p_1(x)}{p_2(x)} \big) \lambda (dx) = \mathbb{E}_{\mu_1}\left[ \log \frac{p_1(x)}{p_2(x)} \right] 
 \end{equation}
 We can also write $K(p_1,p_2)$.
\end{definition}
\par\hspace{1cm} We will focus on in this article to a particular family of \emph{Csisz\'ar $\phi$-divergence} that is the family of \emph{$\alpha$-divergence} measures. Thus we have the $\alpha$-divergence \cite{Cichocki2008}, the R\'enyi $\alpha$-divergence and the Tsallis $\alpha$-divergence of the discrepancy which we will work.
\begin{definition}[$\alpha$-divergence]
\par The $\alpha$-divergence  is defined by
 \begin{equation}
 D_{\alpha}(\mu_1,\mu_2)=\frac{1}{\alpha(1-\alpha)}\Big(1-\int_{\mathcal{X}}p_1^\alpha(x)p_2^{1-\alpha}(x)
 \lambda(\mathrm{d}x)\Big),\, \alpha > 0 \ \hbox{and} \ \alpha \neq 1.
\end{equation}
\end{definition}
\begin{definition}[R\'enyi $\alpha$-divergence]

\par R\'enyi (1961) for the first time gave one generalization of the relative entropy given in \ref{h1}. It is defined by
\begin{equation}
 R_\alpha(\mu_1,\mu_2) =\frac{1}{\alpha-1} \log\int_{\cal X}p_1^{\alpha}(x)p_2^{1-\alpha}(x)\lambda (dx)
,
 \, \alpha > 0 \ \hbox{and} \ \alpha \neq 1.
 \end{equation}

\end{definition} 
 
\begin{definition}[Tsallis $\alpha$-divergence]
\par Tsallis $\alpha$-divergence is defined by
 \begin{equation}
T_\alpha(\mu_1,\mu_2)=\frac{1}{\alpha-1}\Big(\int_\mathcal{X}p_1^\alpha(x)p_2^{1-\alpha}(x)
 \lambda(\mathrm{d}x)-1\Big), \, \alpha > 0 \ \hbox{and} \ \alpha \neq 1
 \end{equation}
\end{definition} 

\hspace{1cm}Chauveau and Vandekerkhove \cite{Chauv2007} studied the strategies of simulation with Kullback-Leibler divergence. Their study was the comparison of simulation strategies by the Kullback-Leibler divergence which allowed to choose the best among  candidate strategies. They used the Metropolis-Hastings (MH) algorithm to explain their method. Thus, with a target density f, proposal density q and an initial distribution $p_i^0$ corresponding to the strategy \textit{i} of MH algorithm they assumed several conditions on these densities to provide certain regularity properties, as the Lipschitz property, on the successive densities $p_i^n$ of strategy \textit{i}. This allowed to obtain an consistent estimate of the entropy
\[ \mathcal{H}(p_i^n)=\int_{\mathcal{X}}p_i^n(x)\log(p_i^n(x))\lambda(\mathrm{d}x) \]

of $p_i^n$ at time $n$.\textsc{}

\hspace{1cm}In this present paper we propose to study how to find the optimal stochastic simulation algorithm  may come from the MH methods (IS and RWMH), Gibbs sampler, or recent adaptive (AM) and hybrid \small{(MWG)} methods. For this, we will use  our three divergence measures. We see that we can find the Kullback-Leibler divergence if for each divergence, the parameter $\alpha$ tends to 1. The interest of  these divergences is that its subtracts the target density $f$, the proposal density $q$ and initial density $p_i^0$  many assumptions like in Chauveau and Vandekerkhove \cite{Chauv2012}. 

\section{Convergence and estimators of divergence measures}
\subsection{Convergence of $\alpha$-divergence measures}
\par\hspace{1cm} We will use a result of Holden (1998) to show the geometric convergence of the MH algorithm under a minoration condition: if there exists $a\in (0,1)$ such that $q(y|x)\geq a f(y)$ for all $x, y\in \mathcal{X}$, then
\begin{equation}\label{h2}
\forall y\in \mathcal{X}, \Big\vert\frac{p^n(y)}{f(y)}-1\Big\vert\leq(1-a)^n\Big\Vert\frac{p^0}{f}
-1\Big\Vert_\infty
\end{equation} 

\begin{theorem}\label{h3} 
If the proposal density of the Metropolis-Hastings algorithm satisfies $q(y|x)\geq \delta f(y)$, for all $x, y\in\mathcal{X}$, and $\delta\in(0,1)$, then with $\alpha\in(1,+\infty)$
\begin{equation}
D_{\alpha}(p^n,f)\leq\frac{1}{\alpha(1-\alpha)}(1-(r\nu^n + 1)^\alpha),
\end{equation}
\begin{equation}
R_\alpha(p^n,f)\leq \frac{\alpha}{\alpha -1}r\nu^n
\end{equation}
\begin{equation}
T_\alpha(p^n,f)\leq\frac{1}{\alpha-1}\big((r\nu^n +1)^\alpha -1\big)
\end{equation}
where $r= \| \frac{p^0}{f}-1 \|_\infty > 0, $ and  $\nu=(1-\delta).$ 
\end{theorem}

\begin{proof}
 {\bf Proof. }
\par \textit{a) $\alpha$-divergence $D_\alpha$} 

\begin{equation}
  D_\alpha(p^n,f)    = \frac{1}{\alpha(1-\alpha)}\Big(1-\int_{\mathcal{X}} \big(\frac{p^n(y)}{f(y)}-1+1\big
)^\alpha f(y)\lambda(\mathrm{d}y) \Big)
\end{equation}

with $\alpha\in]1,+\infty[$, we have 
\[ 1-\int_{\mathcal{X}}\big(\frac{p^n(y)}{f(y)}-1+1\big)^\alpha f(y)\lambda(\mathrm{d}y) \geq
1-\int_{\mathcal{X}}\big(\big\vert\frac{p^n(y)}{f(y)}-1\big\vert+1\big)^\alpha f(y)\lambda(\mathrm{d}y)\]
Therefore  using \eqref{h2} we have 
\begin{eqnarray}
D_\alpha(p^n,f) & \leq & \frac{1}{\alpha(1-\alpha)}\Big(1-\int_{\mathcal{X}}\big(\big\vert\frac{p^n(y)}{f(y)}-1\big\vert +1\big)^\alpha f(y)\lambda(\mathrm{d}y)\Big)\nonumber\cr
                & \leq & \frac{1}{\alpha(1-\alpha)}\Big(1-\int_{\mathcal{X}}(r\nu^n +1)^\alpha f(y)\lambda(\mathrm{d}y) \Big)\nonumber\cr
                & \leq & \frac{1}{\alpha(1-\alpha)}\big(1-(r\nu^n + 1)^\alpha \big)
\end{eqnarray}
\textit{b) R\'enyi $\alpha$-divergence}
\par We will use \eqref{h2} and the same procedure give us
\begin{eqnarray}
R_\alpha(p^n,f) & = & \frac{1}{\alpha -1}\log\int_\mathcal{X}\big(\frac{p^n(y)}{f(y)}\big)^\alpha f(y)\lambda(\mathrm{d}y)\nonumber\cr
                          & \leq & \frac{\alpha}{\alpha -1}\log(r\nu^n +1)\nonumber\cr
                          & \leq & \frac{\alpha}{\alpha -1}r\nu^n
\end{eqnarray}
\textit{c) Tsallis $\alpha$-divergence}
\par Here too we have
\begin{eqnarray}
T_\alpha(p^n,f) & = & \frac{1}{\alpha -1}\Big(\int_\mathcal{X}\big(\frac{p^n(y)}{f(y)}\big)^\alpha f(y)\lambda(\mathrm{d}y)-1\Big)\nonumber\cr
                & \leq & \frac{1}{\alpha -1}\big(\big(r\nu^n +1\big)^\alpha -1)\big)
\end{eqnarray}

\end{proof}

\hspace{1cm}These three divergence measures are all positive, we see that those converge to zero under the conditions of the Theorem \ref{h3}. This allows us to see that the densities $p^n$ of random variables $X_n$ of the Markov chain converge to the stationary distribution $f$ when $n$ goes to infinity.

\subsection{Estimators for the three divergence measures}
\par\hspace{1cm}We first make an estimator of $\alpha$-divergence $D_{\alpha}$, in the same way we can have an estimator of divergences $T_\alpha$ and R\'enyi divergence $R_\alpha$. Suppose $\lambda$ is the Lebesgue measure on $\mathcal{X}=\mathbb{R}^d$. Then we have

\[ D_\alpha(p^n,f)=\frac{1}{\alpha(1-\alpha)}\Big(1-\int\Big(\frac{f(x)}{p^n(x)}\Big)^{
1-\alpha}p^n(x)\mathrm{d}x\Big) \]
\hspace{1cm}We can initially think to write the integral like a mathematical expectancy and apply the method of Monte Carlo integration.
We would have then

\[D_\alpha(p^n,f)=\frac{1}{\alpha(1-\alpha)}\Big(1-\mathbb{E}\Big(\big(\frac{f(X)}{p^n(X)}\big)^{1-\alpha}\Big) \Big)\]

Using the strong law of large numbers we will have the following estimator
\[ \hat{D}_\alpha(p^n,f)=\frac{1}{\alpha(1-\alpha)}\Big(1-\frac{1}{N}\sum_{i=1}^N\big(\frac{f(X_i)}{p^n(X_i)}\big)^{1-\alpha}\Big) \]
which converges almost surely to $D_\alpha(p^n,f)$.
\par\hspace{1cm} But very often encountered in practice, cases where the density $f$ and the densities $p^n$ are not analytically known. There are also cases where $f$ is known to a multiplicative constant: the Bayesian context.
\par\hspace{1cm} It is often difficult to obtain the analytical form of $p^n$. If $f$ is known analytically, it is not useful to estimate it. But let us consider the more general case where $p^n$ and $f$ are not analytically known . Then we use the methods of nonparametric estimation of probability densities. P\'oczos and Schneider \cite{Poczos2012} then proposed in their paper  an estimator of these densities based on \emph{k-NN method}(k nearest neighbor).

\subsubsection{k-NN based density estimator}
\par\hspace{1cm}This is a method of density estimation proposed by Loftsgaarden and Quesenberry \cite{Loftsgaarden1965}. It has a set of data training $E=\left\{x_1,x_2,\ldots,x_N \right\}$  where $x_i$ are independent observations with common probability density function $g$. For a data $z$, the problem is the estimation of a probability density $g$ at $z$. The principle is to find  among $x_i$ the kth nearest neighbor of $z$ considering euclidean distance and we are interested in the distance between the two points ($x_{i,k}$ and $z$). If $x_{i,k}$ and $z$ are points of $\mathbb{R}^d$ this distance $d(z,x_{i,k}) $ allows us to calculate the volume of the hypersphere with center $z$ and radius $r=d(z,x_{i,k})$.
\par The proposed estimator is as follows
\[ \hat{g}_{k,N}(z)=\frac{k/N}{V(Hs(z,d(z,x_{i,k})))}=\frac{k/N}{\pi^{d/2}r^d/\Gamma((d/2) +1)} \]
where $V(Hs(z,d(z,x_{i,k})))=\pi^{d/2}r^d/\Gamma((d/2) +1)$ is the volume of d-dimensional ball around $z\in\mathbb{R}^d$ with radius $r>0$, $\Gamma(.)$ is the Gamma function. Loftsgaarden and Quesenberry (1965) suggested after several experiments to take $k$ equals to $\sqrt{N}$. We will round it if necessary to the nearest whole.

\par  We find in \cite{Poczos2011} theorems of convergence for density's estimators (k-NN estimators), showing the consistency of these estimators.

\begin{theorem}[k-NN density estimators, convergence in probability]

Let $k(n)$ an integer that denotes the $k(n)$th nearest neighbor in the sample of size n. If 
$\lim_{n\rightarrow\infty} k(n)=\infty$, and $\lim_{n\rightarrow\infty}n/k(n)=\infty$, then 
$\hat{g}_{k(n),N}(x)\longrightarrow_p g(x)$ for almost all $x$.
\end{theorem}

\begin{theorem}[Almost sure convergence in sup norm]

 If\\  $\lim_{n\rightarrow\infty}k(n)/\log(n)=\infty$ and $\lim_{n\rightarrow\infty}n/k(n)=\infty$, then 
\[ \lim_{n\rightarrow\infty}\sup_x|\hat{g}_{k(n),N}(x)-g(x)|=0, \quad\textrm{almost surely.}\] 

\end{theorem}

\par\textit{a) Application to $p^n$}
\par Let $X_{1:N}=(X_1,...,X_N)$ be a simulated sample according to the law with density $p^n$ with $X_i$ i.i.d. We choose one $X_i$ on the sample. Let $\rho_{k,i}$ the euclidean distance between $X_i$ and its kth nearest neighbor on the sample. 
\par The estimate of $p^n$ at $X_i$ is
\begin{equation}
\hat{p}_{k,n,N}(X_i)=\frac{k/(N-1)}{V(Hs(X_i,\rho_{k,i}))}
\end{equation}

$ V(Hs(X_i,\rho_{k,i}))=c\,\rho_{k,i}^d $, where $c$ stands for the volume of a d-dimensional unit ball.\par
Hence
\begin{equation}
\hat{p}_{k,n,N}(X_i)=\frac{k}{(N-1)c\,\rho_{k,i}^d}
\end{equation}
\par\textit{b) Application to $f$} 
\par \hspace{1cm}We should also have the i.i.d. simulations according to the density $f$. But the principle in this article is that we don't know how to do i.i.d. simulations directly from $f$. Then we apply the  MCMC simulation algorithms. After several iterations (large $n$) we can begin to recover samples. However, we take the precaution that two successive samples are separated by $n_0$ iterations. We get now the samples $Y_{1:M}=(Y_1,...,Y_M)$ of random variables which are 'independent' with density $f$. Let $\gamma_{k,i}$ the euclidean distance between $X_i$ and its kth nearest neighbor on the sample $Y_{1:M}$. 
\par The estimator of $f$ at $X_i$ is

\begin{equation}
\hat{f}_{k,N}(X_i)=\frac{k}{M\,c\,\gamma_{k,i}^d}
\end{equation}

\subsubsection{Construction of estimators}
\par\hspace{1cm} P\'oczos and Schneider (2011) proposed an estimator of the integral part that is common to the three divergence measures that we have presented. However, their first estimator is asymptotically biased. They showed that multiplying by a constant, they obtained an asymptotically unbiased estimator.
Starting from it we construct our estimators for the $\alpha$-divergence $D_\alpha$, the R\'enyi $\alpha$-divergence $\mathcal{R}_\alpha$ and the Tsallis $\alpha$-divergence $T_\alpha$.

\par\textit{a) Estimator for $\alpha$-divergence $D_\alpha$}
\par\hspace{1cm}Begin by noting  $M_\alpha(p^n,f)=\int (p^n(x))^\alpha (f(x))^{1-\alpha}\mathrm{d}x $ common integral part of the three divergence measures.
\par Once the densities $p^n$ and $f$ estimated at $X_i$ , we have the following estimator of $\hat{D}_{\alpha,N}(p^n,f)$.

\begin{eqnarray}
\hat{D}_{\alpha,N}(p^n,f) & = &\frac{1}{\alpha(1-\alpha)}\Big(1-\frac{1}{N}\sum_{i=1}^N
\Big(\frac{\hat{f}_{k,N}(X_i)}{\hat{p}_{k,n,N}(X_i)}\Big)^{1-\alpha}\Big)\nonumber\cr
& = & \frac{1}{\alpha(1-\alpha)}\Big(1-\frac{1}{N}\sum_{i=1}^N\Big(\frac{(N-1)\rho_{k,i}^d}{M \gamma_{k,i}^d}\Big)^{1-\alpha}\Big)
\end{eqnarray}
We have now
\[\hat{M}_{\alpha,N}(p^n,f)=  \frac{1}{N}\sum_{i=1}^N\Big(\frac{(N-1)\rho_{k,i}^d}{M\gamma_{k,i}^d}\Big)^{1-\alpha} \]
that is an asymptotically biased estimator of $M_\alpha(p^n,f)$
 (P\'oczos and Schneider (2011)). When we multiply it by a constant
\begin{equation}
B_{k, \alpha}=\frac{(\Gamma(k))^2}{\Gamma(k-\alpha +1)\Gamma(k+\alpha -1)} 
\end{equation} 

independent of $ p^n $ and $ f $, with $\Gamma(.)$ the Gamma function. We obtain an asymptotically unbiased estimator.

\par We obtain now a new asymptotically unbiased estimator of $D_\alpha(p^n,f)$

\begin{equation}
\hat{D}_{\alpha ,N,k}(p^n,f)=\frac{1}{\alpha(1-\alpha)}\Big(1-\frac{1}{N}\sum_{i=1}^N\Big(\frac{(N-1)\rho_{k,i}^d}{M\gamma_{k,i}^d}\Big)^{1-\alpha}
\times B_{k, \alpha}\Big)
\end{equation}

\par\textit{b) Estimator for Tsallis $\alpha$-divergence}
\par \hspace{1cm}Applying the same process we obtain an asymptotically unbiaised estimator   $\hat{T}_{\alpha,N,k}(p^n,f)$ of $T_\alpha(p^n,f)$

\begin{equation}
\hat{T}_{\alpha,N,k}(p^n,f)=\frac{1}{\alpha -1}\Big(\frac{1}{N}\sum_{i=1}^N\Big(\frac{(N-1)\rho_{k,i}^d}{M\gamma_{k,i}^d}\Big)^{1-\alpha}
\times B_{k, \alpha}-1\Big)
\end{equation}

\par\textit{c) Estimator for R\'enyi divergence}
\par \hspace{1cm}The R\'enyi measure divergence has a characteristic due to the presence of the logarithmic function that takes $M_\alpha(p^n,f)$ as argument. We have not shown in this paper that its estimator is unbiased  or asymptotically unbiased. However, we propose an estimator 
\begin{equation}
\hat{R}_{\alpha,N,k}(p^n,f)=\frac{1}{\alpha -1}\log\Big(\frac{1}{N}\sum_{i=1}^N\Big(\frac{(N-1)\rho_{k,i}^d}{M\gamma_{k,i}^d}\Big)^{1-\alpha}\times B_{k, \alpha}\Big)
\end{equation}

\par \hspace{1cm}Now give some theorems which show that $\hat{M}_{\alpha,N}(p^n,f)B_{k,\alpha}$ is asymptotically unbiased (P\'oczos and Schneider \cite{Poczos2011}). Define first the following function 

\begin{equation}
F(x,p,\delta,\omega)=\sum_{j=1}^{k-1}\Big(\frac{1}{j!}\Big)^\omega \Gamma(\gamma +j\omega)\Big(
\frac{p(x)+\delta}{p(x)-\delta}\Big)^{j\omega}\big(p(x)-\delta\big)^{-\gamma}\big((1-\delta)\omega\big)^{
-\gamma-j\omega}
\end{equation}

\begin{theorem}[Asymptotic unbiasedness]\label{h4}
 Assume that\\ $0<\gamma=1-\alpha<k$, $p$ is bounded away from 0, $p$ is uniformly Lebesgue approximable, $\exists\delta_0$ such that $\forall\delta\in (0,\delta_0),\:\int F(x,p,\delta,1)p(x)\mathrm{d}x<\infty$,\\
$\int\Vert x-y\Vert^\gamma p(y)\mathrm{d}y<\infty$ for almost all $x\in\mathbb{R}^d$, $\int\int\Vert
x-y\Vert^\gamma p(y)p(x)\mathrm{d}y\mathrm{d}x<\infty$ and that $q$ is bounded from above. Then 
\begin{equation}
\lim_{n,m\rightarrow\infty}\mathbb{E}\Big(\hat{M}_{\alpha ,N}(p^n,f)B_{k,\alpha}\Big)= M_\alpha(p^n,f)
\end{equation}
i.e., the estimator is asymptotically unbiased.

\end{theorem}


The following theorem provide conditions under which estimator $\hat{M}_{\alpha,N}(p^n,f)$ is L2 consistent (P\'oczos and Schneider (2011)).

\begin{theorem}[$L_2$ consistency]
We have the following assumptions: $k\geq 2$, $0<\gamma = 1-\alpha< (k-1)/2$, $p$ is bounded away from 0, $p$ is uniformly Lebesgue approximable, $\exists\,\delta_0$ such that $\forall\,\delta\,\in\,(0,\delta_0)\:\int F(x,p,\delta,1/2)p(x)\mathrm{d}x<\infty$,\\
$\int\Vert x-y\Vert^\gamma p(y)\mathrm{d}y<\infty$ for almost all $x\in \mathbb{R}^d$,
$\int\int\Vert x-y\Vert^\gamma p(y)p(x)\mathrm{d}y\mathrm{d}x<\infty$, and that $q$ is bounded above. Then
\begin{equation}
\lim_{n,m\rightarrow\infty}\mathbb{E}\Big(\big(\hat{M}_{\alpha,N}(p^n,f)B_{k,\alpha}-M_\alpha(p^n,f)\big)^2\Big)=0
\end{equation}
that is, the estimator is $L_2$ consistent.
\end{theorem}



\par\hspace{1cm}The last  theorem show the consistency $ \hat{M}_{\alpha,N}(p^n,f)$ which is the estimator of the common integral part  to the three differences. Based on this, estimators  $\hat{\mathcal{R}}_{\alpha,N,k}(p^n,f)$, $\hat{D}_{\alpha,N,k}(p^n,f)$ and $\hat{T}_{\alpha,N,k}(p^n,f)$ are also consistent.

\section{Asymptotic distribution of divergences estimators}

\hspace{1cm} We seek to know the asymptotic distribution of our estimators. The asymptotic distribution of these estimators is studied under the assumption of continuity of densities $p^n$ and $f$. If densities $p^n$ and $f$ are continuous, their estimators will also be continuous.

\subsection{Asymptotic distribution of $\alpha$-divergence estimator}

Recall that the asymptotically unbiased estimator of

\[ D_\alpha(p^n,f)=\frac{1}{\alpha(1-\alpha)}\Big(1-\int\Big(\frac{f(x)}{p^n(x)}\Big)^{1-
\alpha} p^n(x)\mathrm{d}x\Big) \]
is
\[ \hat{D}_{\alpha,N,k}(p^n,f)=\frac{1}{\alpha(1-\alpha)}\Big(1-\frac{1}{N}\sum_{i=1}^N\Big(\frac{(N-1)V(Hs(X_i,\rho_{k,i}))}{mV(Hs(X_i,\gamma_{k,i}))}\Big)^{1-\alpha}\times B_{k, \alpha}\Big). 
\]
Let
\[h(X_i)=\Big(\frac{(N-1)V(hs(X_i,\rho_{k,i}))}{mV(Hs(X_i,\gamma_{k,i}))}\Big)^{1-\alpha}
\]
and
\[ \bar{Y}_N=\frac{1}{N}\sum_{i=1}^N h(X_i) \]

We know that $X_i$ are i.i.d. $\sim p^n$  . Function $h$ which is the ratio of two continuous functions is also continuous. We have therefore the independence of $h(X_i),\,i=1,\ldots,N$. We can then apply the central limit theorem. Then we will have

\begin{eqnarray}
\frac{\bar{Y}_N-\mathbb{E}(\bar{Y}_N)}{\sqrt{Var(\bar{Y}_N)}}& \xrightarrow{d} & \mathcal{N}\big(0,1 \big)\nonumber\cr
\bar{Y}_N\times B_{k,\alpha}-B_{k,\alpha}\mathbb{E}(\bar{Y}_N)&\xrightarrow{d}& 
\mathcal{N}\Big(0,\sigma^2 B^2_{k,\alpha}\Big)\nonumber\cr
\end{eqnarray}
with $\sigma^2=\lim_{N\to \infty} Var(\bar{Y}_N)$ we have a normal distribution.
But $\bar{Y}_N$ is asymptotically biased (P\'oczos and Schneider \cite{Poczos2011}). Multiplyind it by a constant $B_{k,\alpha}$ as stated by these authors gives us asymptotically unbiased estimator of $M_{\alpha}(p^n ,f)=\int(p^n(x))^\alpha(f(x))^{1-\alpha}\mathrm{d}x$. We will have

\[ \bar{Y}_N B_{k,\alpha}\xrightarrow{d}\mathcal{N}\big(M_{\alpha}(p^n ,f)\,, \sigma^2 B_{k,\alpha}^2\big)
\]

Hence

\begin{equation}
 \hat{D}_{\alpha,N,k}(p^n,f)\xrightarrow{d} \mathcal{N}
\Big(D_{\alpha}(p^n,f)\,, \frac{\sigma^2 B_{k,\alpha}^2 }{\alpha^2(1-\alpha)^2}\Big)
\end{equation}
with  $N,\,M\to\infty$ and $\sigma^2<\infty$.

\subsection{Asymptotic distribution of Tsallis $\alpha$-divergence estimator}

For Tsallis divergence
\[ T_\alpha(p^n,f)=\frac{1}{\alpha -1}\Big(\int\Big(\frac{f(x)}{p^n(x)}\Big)^{1-\alpha}
p^n(x)\mathrm{d}x -1\Big) \]
its estimator
\[
\hat{T}_{\alpha,N,k}(p^n,f)=\frac{1}{\alpha -1}\Big(\frac{1}{N}\sum_{i=1}^N\Big(\frac{(N-1)V(Hs(X_i,\rho_{k,i}))}{M\:V(Hs(X_i,\gamma_{k,i}))}\Big)^{1-\alpha}
\times B_{k, \alpha}-1\Big)
\]
follows a normal distribution. As the previous estimator, we have the same procedure. We have now,

\begin{equation}
\hat{T}_{\alpha,N,k}(p^n,f)\xrightarrow{d}  \mathcal{N}\Big(
T_{\alpha}(p^n,f)\,, \frac{\sigma^2 B^2_{k,\alpha}} {(
\alpha -1)^2}\Big)
\end{equation}

with  $N,M\to\infty$ and $\sigma^2<\infty$.

\subsection{Asymptotic distribution for R\'enyi $\alpha$-divergence estimator}

Now back to our estimator of R\'enyi divergence.

\[ \hat{R}_{\alpha ,N,k}(p^n,f)=\frac{1}{\alpha -1}\log \Big(\frac{1}{N}\sum_{i=1}^N
\Big(\frac{\hat{f}_{k,N}(X_i)}{\hat{p}_{k,n,N}(X_i)}\Big)^{1-\alpha}\times B_{k,\alpha}\Big) \]
We had 
\[ h(X_i)=\Big(\frac{\hat{f}_{k,N}(X_i)}{\hat{p}_{k,n,N}(X_i)}\Big)^{1-\alpha}
\,\textrm{and}\quad \bar{Y}_N =\frac{1}{N}\sum_{i=1}^N h(X_i)\]

$X_i$ are i.i.d. and if $h$ is continuous, then $h(X_i)$ are also i.i.d.

According to the central limit theorem, we will have

\[ \frac{\bar{Y}_N -\mathbb{E}(\bar{Y}_N)}{\sqrt{Var(\bar{Y}_N)}}\xrightarrow{d}
\mathcal{N}(0,1) \]

then,
\[ \sqrt{N}\big(\bar{Y}_N-\mathbb{E}(\bar{Y}_N)\big)\xrightarrow{d}
\mathcal{N}\big(0, N\sigma^2\big) ,\] 

if $\mathbb{E}(\bar{Y}_N)<\infty$, $\sigma^2<\infty$ and $N$ large. We can apply \emph{delta method} and we have

\[ \sqrt{N}\big(\log(\bar{Y}_N)-\log(\mathbb{E}(\bar{Y}_N))\big)\xrightarrow{d}
\mathcal{N}\Big(0, \frac{N\sigma^2}{[\mathbb{E}(\bar{Y}_N)]^2}\Big) .\] 
 
Hence

\begin{equation}
\hat{R}_{\alpha ,N,k}(p^n,f)\xrightarrow{d}\mathcal{N}\Big(\frac{1}{\alpha -1}\log[B_{k,
\alpha}\mathbb{E}(\bar{Y}_N)]\,, \frac{\sigma^2}{(\alpha -1)^2[\mathbb{E}(\bar{Y}_N)]^2}\Big)
\end{equation} 
 
From P\'oczos at Schneider \cite{Poczos2011}, $B_{k,\alpha}\mathbb{E}(\bar{Y}_N)\rightarrow M_{\alpha}(p^n,f)$, when $N, M\rightarrow \infty$. From which 

\begin{equation}
\hat{R}_{\alpha ,N,k}(p^n,f)\xrightarrow{d}\mathcal{N}\Big(R_{\alpha}(p^n,f)\,, \frac{\sigma^2}{(\alpha -1)^2[\mathbb{E}(\bar{Y}_N)]^2}\Big)
\end{equation}

\section{Examples}
\par \hspace{1cm}We will now illustrate our methodology with simple examples. That is why we will limite ourselves to  one-dimensional and two-dimensional cases. In the following examples , we will use  divergence measures to compare proposal densities corresponding to a given simulation strategy. This difference of  densities may appear at their parameters. The comparison will be for different parameters values. The proposal density considered as optimum is that which the divergence measure (function of $n$) between successive densities $p^n$ arising and the target density $f$ tends to $0$ faster.
We will mainly compare as we have already told simulation strategies.
\subsection{One-dimensional case}
\subsubsection{Target density $f$ fully known}
\par \hspace{1cm}In the case where the target density $f$ is analytically known the estimators of our divergence measures are slightly modified. This change applies at the constant $ B_{k,\alpha} $. This constant is replaced by another in order to get there also asymptotically unbiased estimators. Then consider 

\[\hat{M}_{\alpha ,N}(p^n,f)=  \frac{1}{N}\sum_{i=1}^N\Big(\frac{\hat{f}_{k,N}(X_i)}{\hat{p}_{k,n,N}(X_i)}\Big)^{1-\alpha} \]

If we replace the density's estimator $f$ by $f$ itself we will have

\[\hat{M}'_{\alpha,N}(p^n,f)=\frac{1}{N}\sum_{i=1}^N\Big( \frac{f(X_i)(N-1)c\rho_{k,i}^d}{k}\Big)^{1-\alpha} \]
which is an another estimator of the common integrale part $M_\alpha(p^n,f)$ but asymptotically biased. As in the proof of Theorem \ref{h4} (P\'oczos at Schneider (2011)) we can simply check that the new estimator multiplied by the constant 
$Q_{k,\alpha}=\frac{\Gamma(k-\alpha +1)}{k^{1-\alpha}\Gamma(k)}$ provides an asymptotically unbiased estimator of $M_\alpha(p^n,f)$.

\hspace{1cm}In one-dimension we are in the set $\mathbb{R}$, therefore $d=1$. Consequently the d-dimensional hypersphere around $X_i$  boils down to the line segment with $X_i$ middle. The real $c$ represents the length measuring  of a line segment with $X_i$ middle. This segment has 2 as length measuring since the distance from $X_i$ at each end is 1. As suggested by   Loftsgaarden and Quesenberry \cite{Loftsgaarden1965} we will take $k$ equals to the nearest integer of $\sqrt{N-1}$. We have now
\begin{equation}
\hat{M}'_{\alpha,N}(p^n,f)=\frac{1}{N}\sum_{i=1}^N\Big( \frac{2\rho_{k,i}(N-1)f(X_i)}{k}\Big)^{1-\alpha}  
\end{equation}
\par\hspace{1cm}We can now have the asymptotically unbiased estimators  
\[\hat{D}'_{\alpha,N,k}(p^ n,f)=\frac{1}{\alpha(1-\alpha)}\Big(1-\hat{M}'_{\alpha,N}(p^n,f)\times
Q_{k,\alpha}\Big)\]
\[\hat{T}'_{\alpha,N,k}(p^n,f)=\frac{1}{\alpha-1}\Big(\hat{M}'_{\alpha,N}(p^n,f)\times Q_{k,\alpha}
-1\Big)\] 
respectively of $D_\alpha(p^n,f)$ and $T_\alpha(p^n,f)$. For the estimator of the R\'enyi divergence
\[\hat{R}'_{\alpha,N,k}(p^n,f)=\frac{1}{\alpha-1}\log \Big(\hat{M}'_{\alpha,N}(p^n,f)\times Q_{k,\alpha}\Big)\]
 we don't have results regarding the presence or absence of bias. But the most important is the fact that  three divergence estimators $\hat{D}'_{\alpha,N,k}(p^ n,f)$, $\hat{T}'_{\alpha,N,k}(p^n,f)$ and $\hat{R}'_{\alpha,N,k}(p^ n,f)$  are both consistent.

\par\hspace{1cm}First we show an example that compares two proposal densities\footnote{\small{Compare densities is equivalent to compare the respective probability distributions, so we will use interchangeably densities or probability distributions}}  for a given strategy. We choose the Independence Sampler (IS), which is one of the  strategies of the Metropolis Hastings algorithm, to compare these densities.

\par \textit{a) Independence Sampler (IS)}: comparison of proposal densities
\par\hspace{1cm} For a given simulation strategy  the choice of good proposal density is important. Indeed, for achieving  satisfactory simulation results it is important to choice a good proposal density. For simplicity, we consider densities that differ by the value of their parameters. Thus we take  respective distribution densities  $\mathcal{N}(-3,2)$ and $\mathcal{N}(0,3)$; the target density is from normale standard distribution $\mathcal{N}(0,1)$. It is found by looking at  Figure 1 that the dashed curve converges very rapidly to 0 while the solid curve is slow to converge.

\begin{figure}[h]
	\centering
		\includegraphics[scale=0.55]{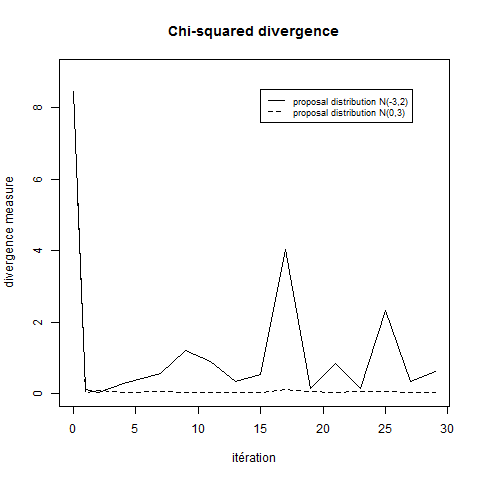}
	\caption{\small{Comparison of two proposal densities,using the Independence Sampler and
	the Chi-squared divergence ($D_{\alpha}$ with $\alpha=2$)}}
\end{figure}
\par \textit{b) IS - RWMH\footnote{Random Walk of Metropolis-Hastings} }: comparison of simulation strategies
\par\hspace{1cm}After finding a good  proposal distribution for each strategy, we compare here the two main strategies of  MH algorithm that are IS and RWMH. It may happen in an experiment that the IS strategy trumps RWMH strategy and in another experiment the opposite occurs. Everything depends on the instrumental distribution but also the target distribution to some extent even if the initial law is the same for both strategies.
\par\hspace{1cm}Here the target distribution is a gaussian mixture $0.4\,\mathcal{N}(-8,2)+0.6\,\mathcal{N}(0,6)$. For the \emph{Independence Sampler} we use the proposal distribution $\mathcal{N}(-2.5,15)$ whereas for RWMH method we propose to take $\mathcal{N}(x,15)$. This distribution has as mean equal to the current element $X_n=x$. We show the comparison of IS and RWMH strategies (Fig. 2). Note that the curve associated with IS strategy is below curve associated with RWMH.
However, the two curves converge very quickly to 0.
\begin{figure}[h]
	\centering
		\includegraphics[scale=0.55]{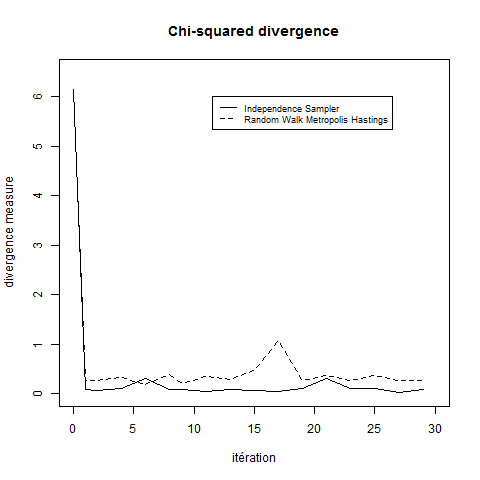}
	\caption{\small{Comparison of two simulation strategies: Independence Sampler vs Random Walk of 
	Metropolis-Hastings, using the Chi-squared divergence ($D_{\alpha}$ with $\alpha=2$) }}
\end{figure}

\subsubsection{Target density $f$ is not known completely}
\par\hspace{1cm} In most real situations, the density $f$ is not known analytically. This is the case, for example Bayesian context where $f$ is the density of the posterior. Then $f$ is written as $f = c\varphi$ where $c$ is unknown constant.     \cite{Poczos2011}
\par \textit{-- Adaptive Metropolis\footnote{This simulation strategy is proposed by Haario et al. (2001)}  - RWMH}
\par Here the target density is known up to a constant
\[ f(x)\varpropto\exp(-x^2)(2+\sin(5x)+\sin(2x)) \]
\par\hspace{1cm}Present some \emph{Adaptive Metropolis} (AM) strategy proposed by Haario et al (2001). First recall that the stochastic process generated by this simulation method is not a Markov chain. However it has well ergodicity properties. The assumptions for this are that the target density is bounded from above and has a bounded support. Describe the algorithm now.
\par\hspace{1cm} The target density has a support $E\subset\mathbb{R}^d$. Suppose, that at time $t$ we have sampled the states $X_0, X_1,\ldots, X_{t-1}$,
where $X_0$ is the initial state. Then a candidate point $Y$ is sampled from the (asymptotically
symmetric) proposal distribution $q_t(.|X_0, ..., X_{t-1})$, which now may depend on the whole
history $(X_0, X_1, ..., X_{t-1})$. The candidate point $Y$ is accepted with probability
\[\alpha(X_{t-1},Y)=\min \Big(1,\frac{\pi(Y)}{\pi(X_{t-1})}\Big)\]
in which case we set $X_t=Y$, and otherwise $X_t=X_{t-1}$. Observe that the chosen probability for the 
acceptance resembles the familiar acceptance probability of the Metropolis algorithm. The proposal
distribution $q_t(.|X_0,\ldots, X_{t-1})$ employed in the AM algorithm is a Gaussian distribution with mean at the current point $X_{t-1}$ and covariance 
\[C_t=\left\lbrace\begin{array}{ll}
C_0\,, & t\leq t_0 \\
S_d Cov(X_0,\ldots,X_{t-1})+S_d \epsilon I_d\,,& t>t_0\,.
\end{array}\right.\] 
where $S_d$ is a parameter that depends only on dimension $d$ and $\epsilon>0$ is a constant that we may choose very small compared to the size of $E$. Here $I_d$ denotes the $d$-dimensional identity matrix.
The covariance $C_t$ may be viewed as a function of $t$ variables from $\mathbb{R}^d$ having values in uniformly positive definite matrices.

\par\hspace{1cm}For these two simulation methods (AM and RWMH) we see that the respective  divergence measures stand very close and very quickly all tend to 0 (Fig. 3). Indeed, near 0 divergence measures which are in this case the \emph{Hellinger divergence} measures $D_{1/2}$ are slight oscillations  and intersect.

\begin{figure}[h]
	\centering
		\includegraphics[scale=0.55]{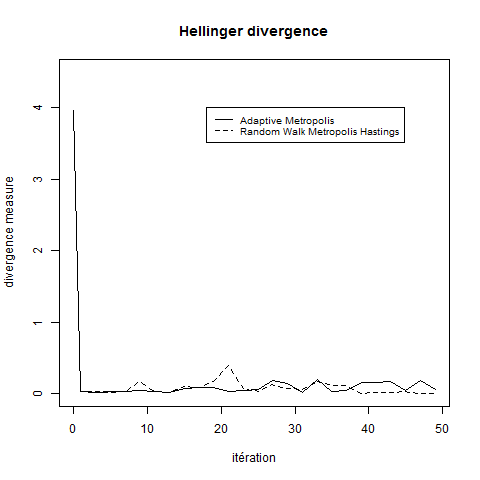}
	\caption{\small{Comparison of two simulation strategies: Adaptive Metropolis vs Random Walk of 
	Metropolis-Hastings, using the Hellinger divergence ($D_{\alpha}$ with $\alpha=1/2$)}}
\end{figure}
\subsection{Two-dimensional case}
\par\hspace{1cm}We consider here the only case of a known function up to a constant. We choose now a sample $X_1,\ldots,X_n$ which are i.i.d.\footnote{independent identically distributed} such that $X_i\sim \mathcal{N}(m,\sigma^2)$.
So we have the following likelihood
\[L(x|m,\sigma^2)\propto(\sigma^2)^{(-n/2)}\exp \Big(-\frac{1}{2\sigma^2}\sum_{i=1}^n(x_i-m)^2\Big),\]
the prior distributions are 
\[m\sim\mathcal{N}(m_0,\sigma_0^2)\]
\[\sigma^2\sim\mathcal{IG}(\alpha,\beta),\]
the full posterior density is known up to a constant
\[\Pi(m,\sigma^2|x)\propto(\sigma^2)^{-\frac{n}{2}-(\alpha+1)}\exp\Big(-\frac{1}{2\sigma^2}
\sum_{i=1}^n(x_i-m)^2-\frac{(m-m_0)^2}{2\sigma_0^2}-\frac{\beta}{\sigma^2}\Big),\]
the conditional distributions of parameters are
\[m|\sigma^2,x\sim\mathcal{N}(M,\Sigma^2)\]
where
\[M=\frac{\sigma_0^2\sum_{i=1}^n x_i +\sigma^2m_0}{\sigma^2+n\sigma_0^2}\quad\textrm{and}\quad
\Sigma^2=\frac{\sigma^2\sigma_0^2}{\sigma^2+n\sigma_0^2}\] 
\[\sigma^2|m,x\sim\mathcal{IG}\Big(\frac{n}{2}+\alpha,\frac{1}{2}\sum_{i=1}^n(x_i-m)^2+\beta\Big)\]
\par\hspace{1cm}So let's compare firstly RWMH and Gibbs sampler and secondly the RWMH and Metropolis Within Gibbs. $\Pi(m,\sigma^2|x)$ will be our target density.
\par\textit{a) RWMH - Gibbs sampler}
\par\hspace{1cm}If RWMH  applied both in dimension $d\geq 1$ , the Gibbs sampler on it only applied in dimension $d> 1$. However, we limit ourselves here in dimension 2. Note that this method (Gibbs sampler) has been used by  Geman (1984) to generate observations from a Gibbs distribution (Boltzmann distribution). It is a particular form of the MCMC method, because of its effectiveness, is widely used in many fields of Bayesian analysis. Thus to simulate according to a probability density $ f (\theta) $ with $ \theta = 
(\theta_1,\ldots,\theta_p)$ one can use the following idea
\par Initialisation: generating a vector $\theta=(\theta_1,\ldots,\theta_p)$ 
according to a initial proposal law $\Pi_0$.  
\par Simulate following the conditional distributions
\[ \theta_i|\theta_1,\ldots,\theta_{i-1},\theta_{i+1},\ldots,\theta_p \sim f_i(\theta_i|
\theta_1,\ldots,\theta_{i-1},\theta_{i+1},\ldots,\theta_p) \]
$i=1,2,\ldots,p$.
\par\hspace{1cm}We see that the curve corresponding to the  RWMH strategy is well above that representing the Gibbs sampler (Fig. 4). The two curves have only one common point that is their origin. As mentioned in the legend, the dashed curve is associated with RWMH while the solid curve is associated with the Gibbs sampler. However, the curves denote divergence measures, we will talk about it in Section 6 (Discussion).
\begin{figure}[h]
	\centering
		\includegraphics[scale=0.55]{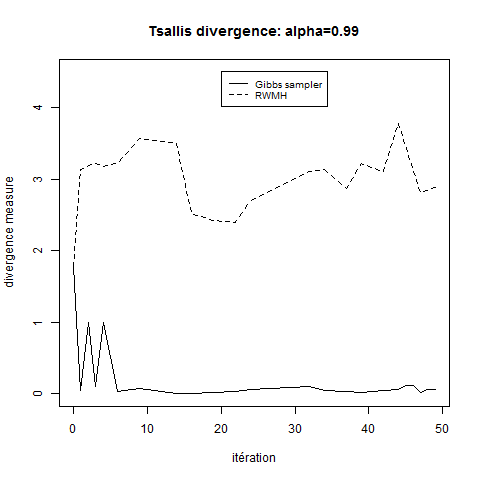}
	\caption{\small{Comparison of two simulation strategies: Gibbs sampler vs Random Walk of 
	Metropolis-Hastings, using the Tsallis divergence  with $\alpha=0.99$}}
\end{figure}

\par\textit{b) RWMH - Metropolis Within Gibbs}
\par\hspace{1cm}\emph{Metropolis Within Gibbs} is a hybrid simulation method that combines stages of the Gibbs sampler and Metropolis Hastings method. It is used in some cases where we have conditional distributions for which we can't have samples directly. There are several versions of this sampler, so we present the following.
\par\hspace{1cm}Assume that $\pi(.)$ is the target distribution\footnote{$\pi(.)$ and $q(.)$ are the probability densities functions} , $\pi(.|x_{-i})$ denote now the conditional distribution of $Z|Z_{-i}=z_{-i}$ where $Z\sim\pi$.
 $X_n:=(X_{n,1},\ldots,X_{n,d})$; $X_{n,-i}:=(X_{n,1},\ldots,X_{n,i-1},X_{n,i+1},X_{n,d})$;
 $\alpha:=(\alpha_1,\ldots,\alpha_d)$.\\
 \par\hspace{1cm}Now we have the following algorithm.
 \par\texttt{Algorithm}\\
 1. Choose coordinate $i\in\lbrace 1,\ldots,d\rbrace $ according to selection probabilities $\alpha$, that      is, with $\mathbb{P}(i = j ) = \alpha_j$ .\newline
 2. Draw $Y\sim q(X_{n-1,i},.)$.\newline
 3. Accepte the candidate $Y$ with probability
 \[\min\Big(1,\frac{\pi(Y|X_{n-1,-i})q(Y,X_{n-1,i})}{\pi(X_{n-1,i}|X_{n-1,-i})q(X_{n-1,i},Y)}\Big)\]
 and set $X_n:=(X_{n-1,1},\ldots,X_{n-1,i-1},Y,X_{n-1,i+1},\ldots,X_{n-1,d})$\\
 otherwise reject $Y$ and set $X_{n}=X_{n-1}$.
\par\hspace{1cm}Starting from a common point, the curves stay away from the value 0 (Fig. 5). The solid curve is associated with \emph{Metropolis  Within Gibbs} strategy and the dashed curve  associated with \emph{RWMH} strategy. Here we use the \emph{R\'enyi divergence} measure in order to obtain our curves. 
 \begin{figure}[h]
	\centering
		\includegraphics[scale=0.53]{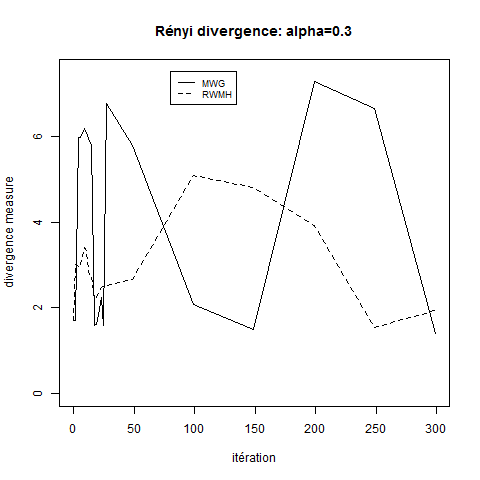}
	\caption{\small{Comparison of two simulation strategies:\emph{ Metropolis within Gibbs} vs \emph{Random Walk of Metropolis-Hastings}, using the R\'enyi $\alpha$-divergence  with $\alpha=0.3$}}
\end{figure}
\section{Discussion}
\subsection{Indepence Sampler : comparison of proposal densities}
\par\hspace{1cm}The two curves represent measurements of \emph{Chi-square divergence} $D_2$. Indeed, the dashed curve indicates the difference between the densities $p^n_1$ from the IS strategy with  proposal law $\mathcal{N} (0,3)$ and the target density $f$ of $\mathcal{N}(0,1)$, noted by $D_2(p^n_1,f)$. The solid curve denotes the difference between the densities $p^n_2$ produced by the IS strategy with proposal distribution $\mathcal{N}(-3,2)$ and the target distribution  $\mathcal{N}(0,1)$, noted by $D_2(p^n_2,f)$. 
\par\hspace{1cm} It is clear that $D_2(p^n_1,f)$ and $D_2(p^n_2,f)$ are functions of the number of iterations $n$ and have the same origin. The fact that they have the same starting point can be explained by  relevance to begin with a same state $x_0$  or same initial law $\pi_0$ to compare two simulation strategies. The two curves overlap up to the first iteration. Indeed, the IS strategy with $\mathcal{N}(0,3)$ is more efficient because this distribution has a mean $m=0$ like the mean of target distribution. Its variance $\sigma^2 = 3$ is greater than the variance of the target distribution which is $\sigma^2 = 1$. It follows that its support covers the support of the target density. The other strategy has a  proposal distribution $\mathcal{N}(-3,2)$. This law has a mean $m = -3$ and a variance $\sigma^2=2$  that make its density function is shifted to the left (see figure) and is therefore not adequate to cover the support of the target density. So that the chain  generated by this simulation strategy  will soon converge to the target distribution. 
\subsection{Independence Sampler - Random Walk of Metropolis Hastings}
\par\hspace{1cm} We said earlier that the two curves in Figure 2 converge rapidly to 0.We use also, here \emph{Chi-square divergence}.Dashed curve represent also the divergence measure   between the densities $p^n$ of RWMH strategy and target density $f(x)=0.4f_1(x,-8,2)+0.6f_2(x,0,6)$ which is Gaussian mixture density. Similarly, the solid curve represents divergence measure between the densities  $p^n$ from the IS strategy and the target density $f(x)$.
\par\hspace{1cm} Although the two divergence measures converge rapidly to $0$, we see that the divergence which represents the IS strategy is below the divergence associated with RWMH strategy. Therefore  IS strategy is more efficient, even if RWMH  is also  good method.
\par\hspace{1cm} Note that the two curves have the same starting point and overlap untill the first iteration. From there, the curve of IS strategy approximates much to $0$. This is due to the fact this method (IS) has, here, a good proposal distribution $\mathcal{N}(-2.5,15)$. With a mean $m=-2.5$ and a variance $\sigma^2=15$, its density function is centered relatively to the target density $f(x)$. Thus the support of the target density is well covered by the proposal density. 
\par\hspace{1cm} Regarding the  RWMH strategy, its has a normal proposal distribution  to each iteration $n$ with a mean equal to the current state $X_n=x_n$ and variance also equal to $\sigma^2=15$. That's why even if its has not a support that covers all the time  the target distribution's support, its does not so far away. Hence, this is also a good strategy. 
\subsection{Adaptive Metropolis - RWMH}
\par\hspace{1cm} For this comparison,  we use the \emph{Hellinger divergence} measure  $(D_{1/2})$. The solid curve then describes the divergence measure between densities $p^n$ (AM) and the target density $f(x)$. The dashed curve described the  divergence between the densities $p_n$ (RWMH)and also  target density.  The two curves merge between the initial state and the first iteration. From there, they don’t move away from each other, but coexist  around value $0$. Thus, the two corresponding simulation strategies are all very efficient. 
\par\hspace{1cm} The effectiveness of these two strategies is explained in first concerns the AM strategy, by the fact that it adapts its proposal density to the target density. The adjustment mechanism is performed at the variance of instrumental density. Indeed, if $X_1, X_2,\ldots, X_{n-1}$ have already been simulated and one desire to obtain the point (or vector) $X_n$ at time $n$, we generate with a proposal distribution with mean equal the current state of $X_{n-1}$ and covariance martice described above.
\par\hspace{1cm} Recall that we use the AM strategy here in one dimension. Thus, the variance $\sigma^2=5$ chosen for our instrumental Gaussian between the initial time and time $t=15$ ( this is the period prior to the iterative update of the variance ) yields samples which describe almost support the target density. Added to this, from the 16th iteration adjustments allow the variance to have good samples. Which explains why it has a very fast convergence.
Concerning RWMH method, with same variance $\sigma^2=5$ for the proposal density, we also obtain good samples that converge very quickly to the stationary density.
\subsection{Gibbs Sampler - RWMH}

\par\hspace{1cm} As we said earlier, we have in Figure 4 two curves designating measures of Tsallis divergence with  $\alpha = 0.99$.  The dashed curve (RWMH) is above the other and far from the value $0$ after the first 50 iterations. It shows, here, that RWMH strategy is ineffective . This inefficiency is due to the covariance matrix $M$ of the proposal distribution. This matrix which is  implementation parameter is not optimum. The solid curve (Gibbs sampler) tends rapidly to 0 (after the 7th iteration). The strength of this sample is mainly due to the fact that the components of the vector $X_n$ are generated directly from the simulated conditional distributions of the target distribution.

\subsection{Metropolis Within Gibbs - RWMH}
\par\hspace{1cm}The R\'enyi divergence measures between the respective densities $p^n$ both strategies (MWG and RWMH) and the stationary distribution are held away from the value 0 even after a large number of simulations.
What makes us to say that  two strategies are not effective.
\par\hspace{1cm} As in the previous comparison (Gibbs Sampler vs RWMH) RWMH  strategy is always slow in computation time. In fact it is the same strategy with the  instrumental Gaussian distribution with parameter matrix $M$, but evaluated using the R\'enyi divergence measure with $\alpha = 0.3$. We see thereby that we have a wide range of divergence measures allowing us to evaluate a given strategy simulation. Respecting the MwG method in our case, even after 30,000 iterations, the created process (Markov chain)  does not converge (Fig 5); even if we are limited here in 1,000 iterations. 
\par\hspace{1cm} This method (MwG) is often used in cases where one wants to use the Gibbs sampler and for some conditional distributions he can not simulate directly. Then one introduce in the algorithm a few steps of  RWMH sampler, offering instrumental distributions having for target laws the conditional distributions. In our case (Fig. 5) we used a version that systematically applies the steps of the Metropolis sampler to all conditional distributions, so that the algorithm will delay to converge.

\subsection{Conclusion}

\par\hspace{1cm} We have, therefore, shown that with various divergence measures we can compare two different simulation strategies. To achieve this, we used $\alpha$-divergence measures  that we have described in Section 2. We showed the convergence of these divergence measures and proposed estimators for these ones (Section 3). In Section 4 we gave the asymptotic distribution of each estimator. Then we gave some examples for the implementation of simulation strategies that have been described (Section 5). Finally we discussed in Section 6 of the causes that make a simulation strategy is more powerful than the other.

\nocite{*}
\bibliographystyle{plain} 
\bibliography{mcmc-1}

\begin{thebibliography}{10}

\bibitem{Cichocki2008}
Y.~D.~Kim A.~Cichocki, H.~Lee and S.~Choi.
\newblock Non-negative matrix factorization with $\alpha$-divergence.
\newblock {\em Pattern Recognition Letters}, 2008.

\bibitem{Mehta2011}
N.~K.~Verma B.~Mehta and P.~Sircar.
\newblock Performance analysis of alpha divergence in nonnegative matrix
  factorization of monaural musical sounds.
\newblock {\em International Journal of Engineering, Science and Technology},
  3(6):273--282, 2011.

\bibitem{Poczos2012}
J.~Xiong B.~P\'oczos and J.~Schneider.
\newblock Nonparametric {D}ivergence {E}stimation with {A}pplications to
  {M}achine {L}earning on {D}istributions.
\newblock {\em arXiv:1202.3758v1 [cs.LG]}, pages 599--608, 2012.

\bibitem{Chauv2007}
D.~Chauveau and P.~Vandekerkhove.
\newblock How to compare mcmc simulation strategies?
\newblock {\em hal-00019174 (version 3)}, 2007.

\bibitem{Chauv2012}
D.~Chauveau and P.~Vandekerkhove.
\newblock Smoothness of metropolis-hastings algorithm and application to
  entropy estimation.
\newblock {\em ESAIM: Probability and Statistics}, 2012.

\bibitem{Cichocki2010}
A.~Cichocki and S.~Amari.
\newblock Families of {A}lpha- {B}eta- and {G}amma- {D}ivergences: {F}lexible
  and {R}obust {M}easures of {S}imilarities.
\newblock {\em Entropy}, pages 1--41, 2010.

\bibitem{Haario2001}
E.~Saksman H.~Haario and J.~Tamminen.
\newblock An adaptive {M}etropolis algorithm.
\newblock {\em Bernoulli}, 7(2):223--242, 2001.

\bibitem{Latusz2013}
G.~O.~Roberts K.~Latuszy\'nski and J.~S. Rosenthal.
\newblock Adaptive {G}ibbs sampler and related {MCMC} methods.
\newblock {\em The Annals of Applied Probability}, 23(1):66--98, 2013.

\bibitem{Loftsgaarden1965}
D.~O. Loftsgaarden and C.~P. Quesenberry.
\newblock A {N}onparametric {E}stimate of a {M}ultivariate {D}ensity
  {F}unction.
\newblock {\em Annals of Mathematical Statistic}, 36(3):1049--1051, 1965.

\bibitem{Millet}
Annie MILLET.
\newblock Méthodes de monte-carlo.
\newblock Universités Paris 7 et Paris 1, Cours de Master 2: spécialité
  Modélisation Aléatoire.

\bibitem{Poczos2011}
B.~P\'oczos and J.~Schneider.
\newblock On the {E}stimation of {D}ivergences.
\newblock In {\em Proceedings of the 14th International Conference on
  Artificial Intelligence and Statistics (AISTATS)}, volume~15, pages 609--617,
  2011.

\bibitem{van_Erven2010}
T.~van Erven and P.~{Harremoës}.
\newblock R\'enyi {D}ivergence and {M}ajorization.
\newblock {\em arXiv: 1001.4448v3 [cs.IT]}, 2010.

\end{thebibliography}
\end{document}